\documentclass[journal,onecolumn,draftclsnofoot,10pt]{IEEEtran}
\usepackage{epstopdf}
\usepackage[draft]{hyperref} 
\usepackage{setspace}
\usepackage{amsmath}
\usepackage{amssymb}
\usepackage{stfloats}
\usepackage{amsthm}
\usepackage{bbm}
\usepackage{amsfonts}
\usepackage{color}
\usepackage{graphicx}
\usepackage{booktabs} 
\usepackage{subfigure}  

\usepackage[flushleft]{threeparttable} 

\usepackage{upgreek}
\usepackage{bm}
\usepackage{nicefrac}

\usepackage{multirow}

\usepackage{cite}
\usepackage{mathtools}
\usepackage{stmaryrd}
\usepackage[mathscr]{euscript}
\usepackage{array}
\usepackage{mathrsfs}

\usepackage{soul}

\newtheorem{theorem}{Theorem}
\newtheorem{lemma}{Lemma}

\begin{document}

\title{Blind-Adaptive Quantizers}

\author{Aman Rishal Chemmala and Satish Mulleti,~\IEEEmembership{Member,~IEEE}
\thanks{The authors are with the Department of Electrical Engineering, Indian Institute of Technology Bombay, Mumbai, India, 400076.\\ Emails: amanrishalch2@gmail.com, mulleti.satish@gmail.com}
}

\markboth{Blind-Adaptive Quantizers}%
{Shell \MakeLowercase{\textit{et al.}}: Greedy Algorithm for Sensor Selection in Heterogeneous Sensor Networks}


\maketitle

\begin{abstract}

Sampling and quantization are crucial in digital signal processing, but quantization introduces errors, particularly due to distribution mismatch between input signals and quantizers. Existing methods to reduce this error require precise knowledge of the input's distribution, which is often unavailable. To address this, we propose a blind and adaptive method that minimizes distribution mismatch without prior knowledge of the input distribution. Our approach uses a nonlinear transformation with amplification and modulo-folding, followed by a uniform quantizer. Theoretical analysis shows that sufficient amplification makes the output distribution of modulo-folding nearly uniform, reducing mismatch across various distributions, including Gaussian, exponential, and uniform. To recover the true quantized samples, we suggest using existing unfolding techniques, which, despite requiring significant oversampling, effectively reduce mismatch and quantization error, offering a favorable trade-off similar to predictive coding strategies.
\end{abstract}


\section{Introduction}
\label{sec:intro}
The digital representation of analog signals is crucial for efficiently processing them on digital platforms. This process typically involves two main steps: (i) Sampling — where the analog signal is captured as discrete values, such as the signal's instantaneous amplitudes at regular intervals \cite{shannon1949communication,nyquist,eldar_2015sampling} or the times when the signal crosses a set threshold \cite{levelcross_kumaresan,lazar2004perfect,lazar2003time}; and (ii) Quantization —which converts these discrete values from potentially infinite precision to a fixed, finite precision by mapping them to a predefined set of levels \cite{jayant1984digital,gersho2012vector}.

Sampling is a reversible process, meaning that the original analog signal can be reconstructed from its discrete representation, provided certain conditions are met. A well-known example of this is the Shannon-Nyquist sampling theorem, which states that bandlimited signals can be perfectly reconstructed from uniform samples taken at the Nyquist rate \cite{shannon1949communication,nyquist}. In contrast, quantization is a lossy process, and reconstructing the original signal from quantized samples always introduces errors. Generally, these errors can be minimized by increasing the number of quantization levels or the sampling rate, which introduces correlation among the samples and reduces the error. However, both approaches increase the system's bit rate, which is often undesirable in practice.

For a given bit depth and sampling rate, the quantization error is influenced by the choice of quantization levels and the probability distribution of the signal's amplitudes. For instance, with a uniform quantizer that has equally spaced levels, if the signal's large amplitude values are rare, most quantization levels go unused. Alternatively, using a non-uniform quantizer, where levels are more densely packed in regions of lower amplitude, reduces the average quantization error because low-amplitude samples, which are more common, have smaller errors, while the less frequent high-amplitude samples may have larger errors. This type of non-uniform quantization is often implemented using a nonlinear circuit known as a compander, followed by a uniform quantizer \cite{haykin2008communication}. However, the optimality of such companders, like $\mu$-law and $A$-law quantizers, is not always guaranteed for a given signal distribution.

Several methods have been proposed to design optimal non-uniform quantizers (see \cite{gray1998quantization} for a detailed review). For example, the Lloyd-Max algorithm \cite{max1960quantizing,lloyd1982least} iteratively determines optimal quantization levels for a given distribution. Although effective, these methods require precise knowledge of the input distribution. Any mismatch between the actual distribution and the one for which the quantizer is optimized can lead to significant errors \cite{gray1975quantizer}.

Another approach to reducing quantization error is oversampling, a technique also known as predictive coding \cite{gray1998quantization}. Oversampling—sampling at a rate higher than the Nyquist rate—introduces redundancy or correlation among samples, enabling better approximation of each sample by its predecessors. Consequently, instead of quantizing the raw samples, quantizing the difference between the true sample and its prediction reduces the error, as the difference typically has a smaller dynamic range than the original signal. This concept underlies methods such as differential pulse code modulation and delta modulation.

A different approach involves reducing the dynamic range using a modulo-folding operation, which enhances the predictability of samples \cite{ericson_moduloPCM,uls_tsp,uls_romonov,ordentlich2018modulo,bhandari2021unlimited,eyar_tsp,romanov2021blind,weiss2022blind}. In this technique, instead of directly quantizing the signal samples, the signal is first folded, and then the folded signal samples are quantized. While many studies focus on sampling and reconstructing from folded samples without considering quantization \cite{uls_romonov,bhandari2021unlimited,eyar_tsp}, others, such as \cite{ericson_moduloPCM,uls_tsp,ordentlich2018modulo,romanov2021blind,weiss2022blind}, examine how modulo-folding affects quantization error. For instance, \cite{ordentlich2018modulo} explores modulo-folding-based quantization and derives error bounds assuming the input signal's second-order statistics are known. This assumption was later relaxed in \cite{romanov2021blind,weiss2022blind}, where the algorithms do not require prior knowledge of input statistics. Despite these advances, whether the distribution of folded samples aligns well with the quantizer remains unclear.

In this paper, we address the issue of distribution mismatch between the input signal and the quantizer. We propose a method to minimize this mismatch without needing knowledge of the input's distribution. This approach, which we term blind and adaptive, is effective across a wide range of distributions. To mitigate the mismatch error, we introduce a nonlinear transformation involving an amplifier and a modulo-folding operation. This transformation precedes a uniform quantizer. Theoretically, we demonstrate that with sufficiently large amplification, the output distribution of the modulo-folding operation approaches uniformity, aligning closely with the quantizer's characteristics. This reduction in mismatch is achieved independently of the input's distribution. In particular, we showed that scaling with modulo results in a uniform distribution for commonly used distributions, such as Gaussian, exponential, and uniform input. Due to the many-to-one nature of modulo-folding, we suggest employing the existing unfolding methods \cite{uls_tsp} to recover the true quantized samples. The unfolding requires considerable oversampling, which is the cost of reduced mismatch and lower quantization error, as in most predictive coding schemes. Although unfolding necessitates significant oversampling, this trade-off yields decreased mismatch and lower quantization error, akin to predictive coding methods.


In the next section, we formulate the problem. The proposed solution is discussed in Section III, followed by conclusions.

\section{Problem Formulation}
    
Consider a continuous-time finite energy bandlimited signal $x(t)$ whose Fourier transform vanishes outside the frequency interval $[-\omega_M, \omega_M]$. The signal can be sampled at a rate $\omega_s =  \frac{2\pi}{T_s} \geq 2\omega_M$ to have discrete measurements $\{x[n] = x(nT_s)\}_{n \in \mathbb{Z}}$ and can be reconstructed back as
    \begin{equation}
        \label{eq:Shannon_Reconstruction}
        x(t) = \sum_{n=-\infty}^{\infty}x[n]\, \frac{\sin(\omega_M(t-nT_s))}{\omega_M(t-nT_s)}.
    \end{equation}
Post-sampling, the samples are quantized. To this end, consider an $N$-level quantizer $\mathcal{Q}(\cdot)$ defined as
\begin{align}
    \mathcal{Q}(x[n]) = b_i, \quad \text{for } a_{i} < x[n] \leq a_{i+1},
    \label{eq:quantizerQ}
\end{align}
where $\{a_i\}_{i=0}^{N+1}$ are the decision levels and $\{b_i\}_{i=0}^{N}$ are the corresponding representation levels. The step sizes $\{\Delta_i\}_{i=0}^N$ are given by $\Delta_i = a_{i+1} - a_i$.


Quantization always results in error, and the normalized mean square error (NMSE) is one of the standard metrics for quantifying the error. The NMSE is measured as
    \begin{equation}
        \label{eq:NMSE_quantization}
        {\epsilon}_{\mathcal{Q}}(x) = \frac{\sum_{n=-\infty}^{\infty}\left| x[n] - Q(x[n]) \right|^2}{\sum_{n=-\infty}^{\infty}\left|x[n]\right|^2 }.
    \end{equation}
The measure is typically used for deterministic signals. A more useful measure is one that considers the amplitude distribution of the samples, such as
\begin{equation}
    E_{\mathcal{Q}}^{(r)}(x) = \int_{-\infty}^{\infty}\left| u-\mathcal{Q}(u) \right|^r f_X(u) du,
    \end{equation}
where $f_X(u)$ is the probability density function (PDF) of the samples, and $r$ is a positive integer. Note that the error is a function of the representation levels, the decision levels of the quantizer, and the amplitude distribution. The Lloyd-Max algorithm \cite{max1960quantizing,lloyd1982least} estimates the representation levels such that the error is minimized for $r=2$ and a given distribution.

When the quantizer is optimized for a given distribution, but the signal to be quantized follows another distribution, the error is higher. To elaborate on this point, consider two random variables, $X$ and $Y$, with different distributions. If an arbitrary quantizer $\mathcal{Q}$ is used to quantize the samples from the distributions, the corresponding errors are related as \cite{gray1975quantizer} 
\begin{equation}
    \label{eq:bound_on_diff_in_distortion}
    \left| E_{\mathcal{Q}}^{(r)}(X)^{\frac{1}{r}} - E_{\mathcal{Q}}^{(r)}(Y)^{\frac{1}{r}} \right| \leq W_r(X,Y),
\end{equation}
where $W_r(X,Y)$ is the $r$-th Wasserstein distance between random variables $X$ and $Y$ \cite[Ch 2]{panaretos2020wasserstein}. The distance is measured as 
$W_r(X,Y) = \left( \int_0^1 \left| F_X^{-1}(u) - F_Y^{-1}(u) \right|^r du \right)^{\frac{1}{r}},$
where $F_X^{-1}$ and $F_Y^{-1}$ are the inverse cumulative distribution functions (quantile functions) of $X$ and $Y$, respectively. 


Next, consider the quantizer $\mathcal{Q}_Y$ that is optimized for the PDF of \(Y\), then it follows that $E_{\mathcal{Q}_Y}^{(r)}(X)^{\frac{1}{r}} \geq    E_{\mathcal{Q}_Y}^{(r)}(Y)^{\frac{1}{r}}$. Specifically, from  \eqref{eq:bound_on_diff_in_distortion}, we have that 
\begin{equation}
    \label{eq:bound_on_distortion}
    E_{\mathcal{Q}_Y}^{(r)}(X) \leq  \left( E_{\mathcal{Q}_Y}^{(r)}(Y)^{\frac{1}{r}} + W_r(X,Y) \right)^{r}.
\end{equation}
The inequality shows that the farther the distributions are from each other, with respect to the Wasserstein distance measure, the larger the quantization error due to the mismatch. To get a sense of the extent of error due to the mismatches, we consider some examples of common distribution. In the following, we always consider $r=2$.


Consider a uniform quantizer with \( N \) levels within the range \(-\frac{1}{2}\) to \(\frac{1}{2}\). Let \( Y \) be a random variable uniformly distributed in this range, and \( X \) be a Gaussian random variable with zero mean and standard deviation \(\sigma\). The quantization error for \( Y \), denoted by \( E_q^{(2)}(Y) \), is given by
    \begin{equation}
    \label{eq:error_uniformdist_uniformquantizer}
        E_{\mathcal{Q}_Y}^{(2)}(Y) = \frac{\Delta^2}{12} = \frac{1}{12N^2}.
    \end{equation}
To find the bound in mismatch error, $W_2(X, Y)$ is given as $\sigma^2 + \frac{1}{12} - \frac{\sigma}{\sqrt{\pi}}$ \cite{gray1975quantizer}. By using the distance, the error bound is given as
    \begin{equation}
        E_{\mathcal{Q}_Y}^{(2)}(X) \leq \sigma^2 + \frac{1}{12} - \frac{\sigma}{\sqrt{\pi}} + \frac{1}{12N^2} + \frac{2\sqrt{\sigma^2 + \frac{1}{12} - \frac{\sigma}{\sqrt{\pi}}}}{N\sqrt{12}}.
    \end{equation}
For small values of $\sigma$ compared to the uniform quantizer's dynamic range, the signal's values do not cover all the representation levels and, hence, result in a large error. On the other hand, for large $\sigma$, several signal samples fall outside the quantizer's dynamic range and hence are clipped. This again results in a large error. 



Since the error is a function of the  Wasserstein distance, Table~\ref{tab:Quantization_Wasserstein}, summarizes the distance for commonly used distributions (cf. Appendix~\ref{app:SWD}). To further analyze this issue, we simulate two quantizers with range $[-5,5]$, $\mathcal{Q}_{\mathcal{U}}$ and $\mathcal{Q}_\mathcal{N}$. Quantizer $\mathcal{Q}_{\mathcal{U}}$ is a uniform quantizer while $\mathcal{Q}_{\mathcal{N}}$ is a quantizer for the distribution $\mathcal{N}(0, 1)$ \footnote{$\mathcal{N}(\mu, \sigma)$ denotes Gaussian distribution with zero mean and variance $\sigma^2$. $\mathcal{U}[-\lambda, \lambda]$ denotes uniform distribution over the interval $[-\lambda, \lambda]$.}. The quantization errors are listed in Table \ref{tab:simulation_mismatch}. We infer that the mismatch cannot be overlooked. It is unrealistic to expect the optimized quantizers to perform well on signals that differ significantly from the distribution for which they were designed.

One of the simplest ways to bridge the mismatch is to use a transformation that converts $X$'s distribution to that of $Y$. Specifically, a transformation $G(\cdot)$ should be designed such that $G(X) \sim f_Y$. For example, a transformation that can convert an input signal with a CDF $F_X$ to match a desired signal CDF $ F_Y $ is the inverse quantile method \cite[Ch 7]{edition2002probability} given as 
\begin{align}
    G(X) = F_Y^{-1}(F_X).
\end{align}
This approach solves the problem; however, it requires the precise knowledge of $X$. Moreover, $G$ has to be modified when $X$ changes.

The problem we considered in this work is mitigating the mismatch error that does not require knowledge of $f_X$. We will discuss a solution in the following section.


\begin{table}[!t]
\vspace{-1.0em}
  \begin{center}
    \caption{Table of bounds on quantization errors between various probability distributions}
    \label{tab:Quantization_Wasserstein}
    \renewcommand{\arraystretch}{1}
    \vspace{0.5em}
    \begin{tabular}{c c c} 
    \toprule[1.5pt]
      {\(f_X\)} & {\(f_Y\)} & \(  W_2(X,Y) \)\\
      \midrule
      \(\mathcal{N}(\mu_1, \sigma_1)\) & \(\mathcal{N}(\mu_2, \sigma_2)\) &  \(\sqrt{(\mu_2 - \mu_1)^2 + (\sigma_2 - \sigma_1)^2}\)\\
      \(\exp(p_1)\) & \(\exp(p_2)\) & \(\left| \frac{1}{p_1} - \frac{1}{p_2} \right|  \) \\
          \(\mathcal{U}[0,C]\) & \(\exp(p)\) & \(\sqrt{ \frac{2}{p^2} -\frac{3C}{2p} + \frac{C^2}{3}}\)\\
    \bottomrule[1.5pt]
    \end{tabular}
  \end{center}
\vspace{-1.0em}  
\end{table}
\begin{table}[t]
\vspace{-1.0em}
  \begin{center}
    \caption{Quantization error (NMSE) for \(\mathcal{Q}_{\mathcal{N}}\) and \(\mathcal{Q}_{\mathcal{U}}\) with different probability distributions}
    \renewcommand{\arraystretch}{1}
    \label{tab:simulation_mismatch}
    \vspace{0.5em}
    \begin{tabular}{c l c c} 
    \toprule[1.5pt]
      \textbf{Distribution} & \textbf{Parameters} & {\(\epsilon_{{\mathcal{Q}_{\mathcal{N}}}}\)} in dB & \(\epsilon_{\mathcal{Q}_{\mathcal{U}}}\) in dB\\
      \midrule
      \multirow{3}{*}{Normal} & \(\mu\) = 0, \(\sigma\) = 1 & -43.9 & -38.6\\ %
      & \(\mu\) = 0, \(\sigma\) = 0.5 & -39.8 & -32.2\\ %
      & \(\mu\) = 0, \(\sigma\) = 2 & -23.3 & -26.2\\
      \midrule
      \multirow{3}{*}{Uniform} & \(a\) = -5, \(b\) = 5 & -20.0 & -48.1\\
      & \(a\) = -3, \(b\) = 3 & -44.8 & -43.6\\ 
      & \(a\) = -1, \(b\) = 1 & -41.0 & -33.7\\
      \midrule
      \multirow{3}{*}{Exponential} & \(p\) = 1 & -31.4 & -21.8\\
      & \(p\)  = 0.5 & -38.1 & -34.5\\
      & \(p\)  = 2 & -10.0 & -10.9\\
      \midrule
      \multirow{3}{*}{Lognormal} & \(\mu\) = 0, \(\sigma\) = 1 & -6.08 & -6.43\\ %
      & \(\mu\) = 0, \(\sigma\) = 0.5 & -30.7 & -32.5\\ %
      & \(\mu\) = 0, \(\sigma\) = 2 & -0.09 & -0.09\\
    \bottomrule[1.5pt]
    \end{tabular}
  \end{center}
\vspace{-1.0em}  
\end{table}
\noindent

\section{Proposed Blind-Quantizer}
In this section, we discuss the proposed approach that would overcome the shortcomings of the existing methods. We consider samples $x[n]$ of a bandlimited signal $x(t)$ where the sampling rate is greater than the Nyquist rate. The PDF of the samples is $f_X(x)$. Next, we consider a $N$-level quantizer optimized for a random variable $Y$, uniformly distributed between $[-\lambda, \lambda]$ for a known $\lambda>0$. The quantizer is denoted as $\mathcal{Q}_{Y,\lambda}(\cdot)$. We discuss a non-linear transformation that will modify the random variable $X$ such that the distribution matches that of $Y$. To this end, we consider a scaling operation followed by a modulo operation as the transformer. Specifically, the transformer is defined as
\begin{align}
    x_{a,\lambda}[n] = \mathcal{M}_{a,\lambda}(x[n])
     = (a \,x[n]+\lambda)\,\, \text{mod}\,\, 2\lambda -\lambda.
    \label{eq:scale_modulo}
\end{align}
where $a>0$ is the scaling parameter. Note that $x_{a,\lambda}[n] \in [-\lambda, \lambda]$.

We show that, for a fixed $\lambda$ and a sufficiently large $a$, PDF of $ x_{a,\lambda}$ is approximately uniform in the interval $[-\lambda, \lambda]$ for different distributions $f_X(x)$. Hence, $\mathcal{Q}_{Y,\lambda}(\cdot)$ is an optimal quantizer. Before proceeding further, we first discuss the process of recovering the samples $x[n]$ from the quantized version of $x_{a,\lambda}[n]$.

Unlike the conventional transformation used prior to the quantization \cite{gray1998quantization}, the folding operation is not one-to-one. Hence, correlation among the samples is used to estimate $x[n]$ from the quantized samples, as in predictive coding schemes. In \cite{uls_tsp,uls_romonov,eyar_tsp}, the authors have proposed various algorithms for unfolding, that is, to estimate $x[n]$ from $x_{a,\lambda}[n]$. These algorithms consider sampling above the Nyquist rate and then use redundancy among the samples for unfolding. Further, in \cite{uls_tsp}, the authors have shown that with sufficient oversampling, one can estimate $x[n]$ from the quantized and folded samples. We shall rely on these algorithms for the unfolding process and focus on the distribution of the folded samples.


\subsection{PDF of Modulo-Folded Random Variable}
In the following, we derive a random variable's PDF when modulo folded. For the random variable $X$ with PDF $f_X$, let the PDF of the folded random variable $X_{a, \lambda}$ be denoted as $f_{X_{a,\lambda}}$. Then, the cumulative distribution of $X_{a, \lambda}$ is given as
\begin{align}
    \mathrm{Pr}(X_{a, \lambda} \leq \theta)  = \mathrm{Pr} \left( \mathcal{M}_{a,\lambda}(X) \leq \theta\right), \quad \theta \in [-\lambda, \lambda).
    \label{eq:Pr_X_mod_definition}
\end{align}

By using the modulo properties, \eqref{eq:Pr_X_mod_definition} is written as
\begin{equation}
    \label{eq:Pr_X_mod_simplified}
   \mathrm{Pr}(X_{a, \lambda} \leq \theta) = \sum_{m=-\infty}^{\infty} 
    \hspace{-0.1in}\mathrm{Pr} \left( (2m-1)\frac{\lambda}{a} \leq X \leq \frac{2m\lambda}{a} + \frac{\theta}{a} \right),  \quad \theta \in [-\lambda, \lambda). \nonumber
\end{equation}
Expressing this in terms of the cumulative distribution functions (CDF) \(F_X\) for \(X\) and $F_{X_{a, \lambda}}$ for $ X_{a, \lambda} $, we obtain
\begin{align}
    \label{eq:cdf_X_mod}
    F_{X_{a, \lambda}} (\theta) = \sum_{m=-\infty}^{\infty} 
    {F}_{X}\left( \frac{2m\lambda}{a} + \frac{\theta}{a} \right) - {F}_{X}\left((2m-1)\frac{\lambda}{a}\right), 
\end{align}
where $\theta \in [-\lambda, \lambda)$. For $\theta < -\lambda$, we have $F_{X_{a, \lambda}}(\theta) = 0$, and for $\theta \geq \lambda$, we have $F_{X_{a, \lambda}}(\theta) = 1$.
By applying derivative $f_{X_{a, \lambda}} (\theta)$ is given as
\begin{align}
    \frac{\mathrm{d}}{\mathrm{d}\theta} \left( \sum_{m=-\infty}^{\infty} 
    {F}_{X}\left(\frac{2m\lambda}{a} + \frac{\theta}{a} \right) - {F}_{X}\left((2m-1)\frac{\lambda}{a}\right) \right).\nonumber
\end{align}
The conditions for interchanging the summation and derivative are outlined in \cite[Ch 7]{rudin1964principles}. These conditions are satisfied by well-behaved distributions (cf. Appendix~\ref{sec:interchanging _sum_and_derivative} for the proof). Specifically, all distributions for which $f_X(\theta) <\infty$ and $f_X(\theta)$ tend to zero as $|\theta| \rightarrow \infty$. All the distributions considered in this paper, including Gaussian, exponential, and uniform, are well-behaved.

Thus, for well-behaved distributions, we have
\begin{align}
    \label{eq:pdf_X_mod}
    f_{X_{a, \lambda}} (\theta) = \frac{1}{a} \sum_{m=-\infty}^{\infty} 
    {f}_{X}\left(\frac{2m\lambda}{a} + \frac{\theta}{a}\right).
\end{align}

We show that for a fixed $\lambda$ and sufficiently large $a$ compared to the standard deviation of $X$, the modulo-folded random variable approximately follows a uniform distribution. We need a metric to compare two distributions to quantify this approximation. Although the second Wasserstein distance has been used to quantify mismatch error, it involves the inverse CDF, which is difficult to compute. In contrast, the first Wasserstein distance has a solution based on the CDF, which can be easily calculated using \eqref{eq:cdf_X_mod}. Therefore, we will use the first Wasserstein distance $ W_1(X, Y) $, also known as the Earth Mover's Distance, given by
\begin{equation}
W_1(X, Y) = \int_{-\infty}^{\infty} |F_X(\theta) - F_Y(\theta)| \, d\theta. \label{eq:EMD}
\end{equation}
In the following, we will examine how modulo-folding affects several commonly used PDFs. 

\subsection{Gaussian Distribution}
Let $X$ follows a Gaussian distribution, $f_X(\theta)=\frac{1}{\sqrt{2\pi \sigma^2}}e^{-\frac{1}{2}\left( \frac{\theta - \mu}{\sigma}\right)^2}$.
For such a random variable $X$, by using \eqref{eq:pdf_X_mod}, the PDF of the modulo folded variable \({f}_{X_{a,\lambda}}\) is given as
\begin{align}
    \label{eq:f_mod_gaussian}
    f_{X_{a, \lambda}} (\theta) = \sum_{m=-\infty}^{\infty}\frac{1}{a\sigma\sqrt{2\pi}} 
  e^{\left( -\frac{1}{2}\frac{(2m\lambda+\theta)-a\mu}{a\sigma}\right)^{\!2}\,}.
\end{align}
 In Appendix~\ref{subsec:modulo_folded_gaussian_dist}, we showed that as $a$ increases, the quantity $\frac{\lambda}{a\sigma}$ decreases, and consequently the PDF of $f_{X_{a,\lambda}}$ tends towards $\mathcal{U}[-\lambda, \lambda]$. In Fig.~\ref{fig:f_X_lambda_gaussian}, we plotted the $f_{X_{a, \lambda}} (\theta)$ (cf. \eqref{eq:f_mod_gaussian}) for $\lambda = 1$, $\mu = 0$, and $\sigma=1$. As $a$ increases, the curve increasingly resembles a straight line parallel to the $y$-axis and hence, approximates the uniform distribution.
 
\begin{figure}[!t]
    \centering
    \includegraphics[width=2.5in]{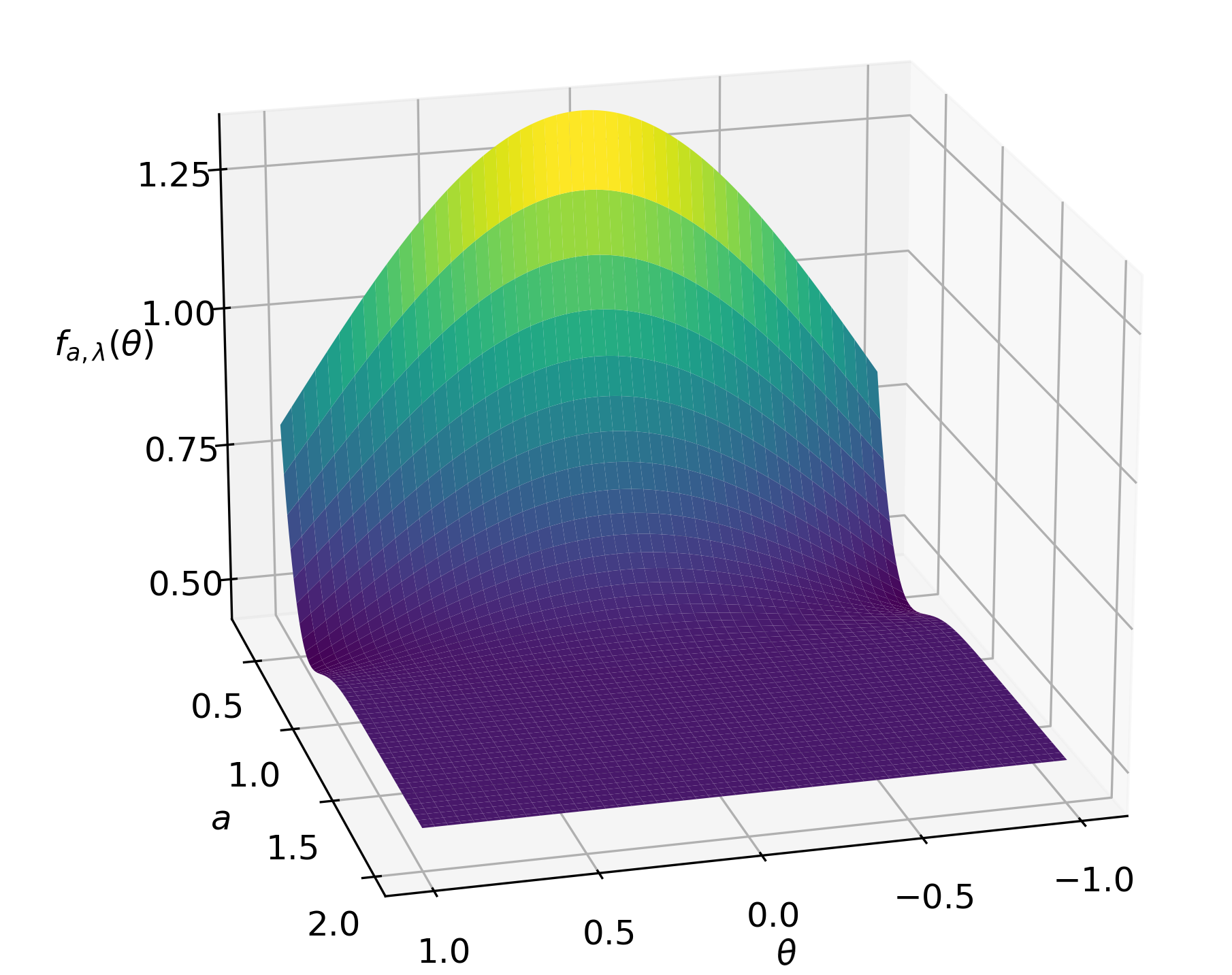}
    \caption{Variation of the PDF $f_{a,X_{\lambda}}(\theta)$ with $a$ and $\theta$ where $X \sim \mathcal{N}[0,1]$.}
    \label{fig:f_X_lambda_gaussian}
\end{figure}
Next, we quantify how well the modulo-folded Gaussian variable approximates a uniformly distributed signal. In Fig.~\ref{fig:vb_example} we showed $W_1$ (see \eqref{eq:EMD}) for different values of $a$ as $\mu$ and $\sigma$ varies. We note that for a given fixed value of $\mu$ and $\sigma$, the approximation error decreases as $a$ increases.

\begin{figure}[!t]
    \centering
    \subfigure[$\mu = 0$]{\includegraphics[width = 3.2 in]{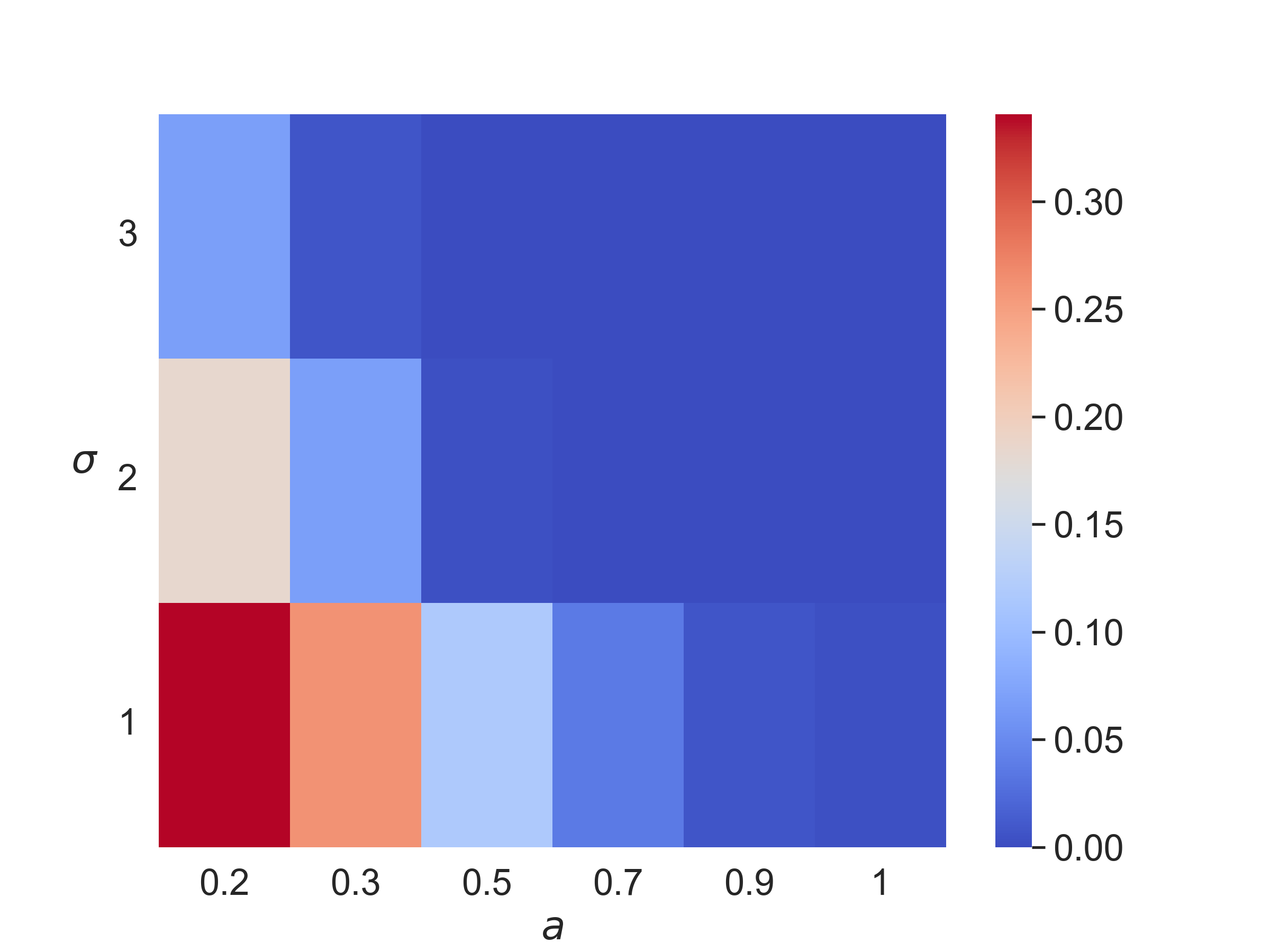}
    \label{fig:heatmap_wasserstein_gaussian_SD_a}
    }
    \subfigure[$\sigma = 1$]{\includegraphics[width = 3.2 in]{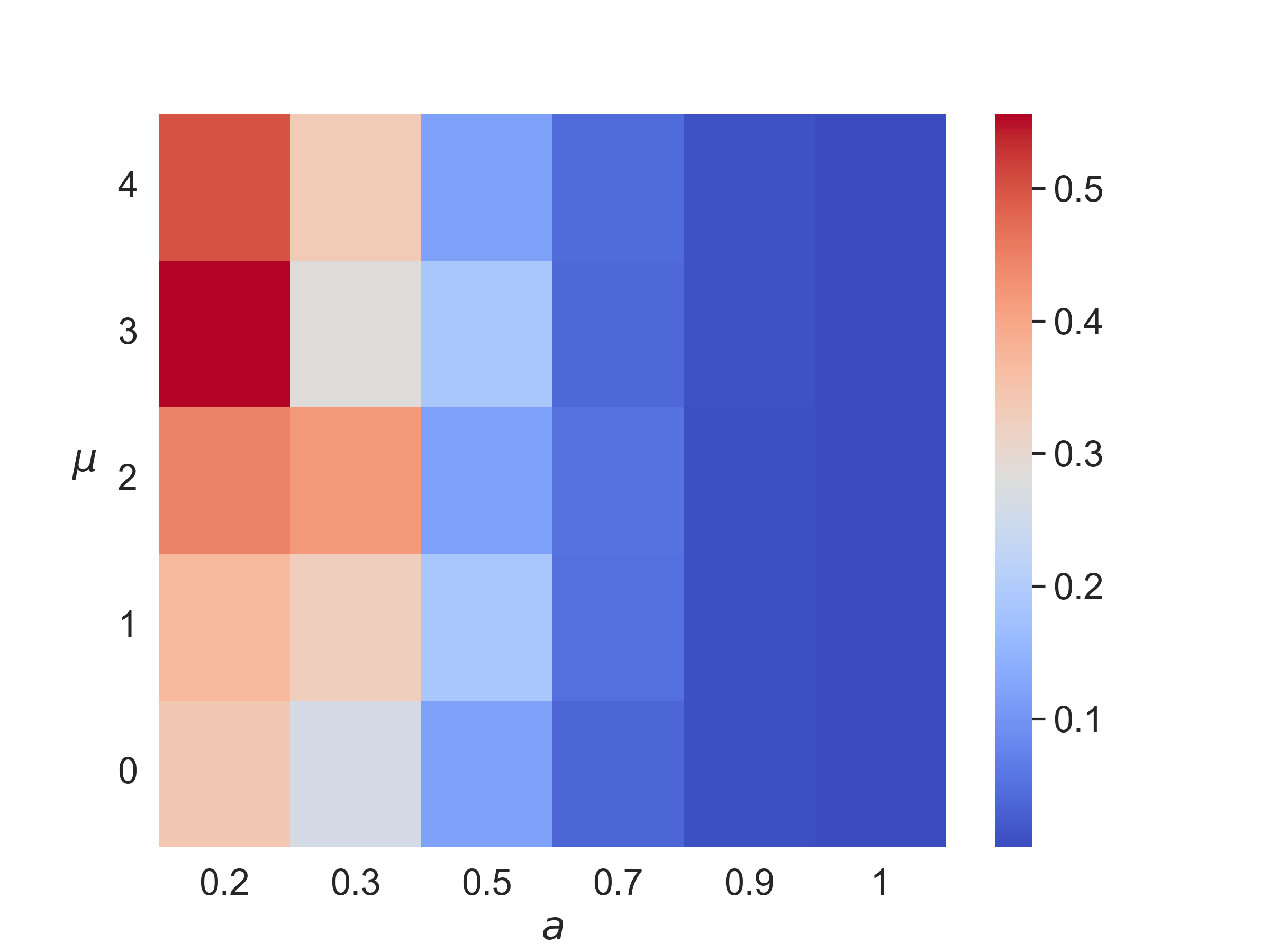}
           \label{fig:heatmap_wasserstein_gaussian_mean_a}
           }
    \caption{The Earth Mover's Distance between a family of Gaussian random variables and $\mathcal{U}[-\lambda,\lambda]$ for $\lambda = 1$. The error decreases as $a$ increases.}
    \label{fig:vb_example}
\end{figure}

\subsection{Exponential Distribution}
An exponential random variable $X$ with rate parameter $p > 0$ has a CDF given by $F_X(x) = 1- e^{-px}$ for $x \geq 0$ and $F_X(x) = 0$ for $x < 0$. In this case, the CDF of $F_{X_{a,\lambda}}$ can be found using \eqref{eq:cdf_X_mod} as
\begin{align}
       F_{X_{a, \lambda}} (\theta) = \begin{cases}
        \frac{e^{\frac{p}{a}\lambda} - e^{-\frac{p}{a}\theta}}{e^{2\frac{p}{a}\lambda}-1} + 1 - e^{-\frac{p}{a}\theta}, &  0 \leq \theta,\\
        \frac{e^{\frac{p}{a}\lambda} - e^{-\frac{p}{a}\theta}}{e^{2\frac{p}{a}\lambda}-1}, & \theta < 0.
    \end{cases}
\end{align}
The derivation details are presented in Appendix~\ref{subsec:modulo_folded_exponential_dist}.
Further, in the appendix, we showed that for large $a$ or small $\frac{p\, \lambda}{a}$, PDF of $X_{a, \lambda}$ converges $\mathcal{U}[- \lambda, \lambda]$.  
Note that $\frac{1}{p}$ represents the standard deviation of the exponential distribution. Thus, the convergence results are analogous to the Gaussian case in the sense that the approximation improves as $a$ increases compared to the standard deviation of the random variable.


\subsection{Uniform Distribution}
Here, we consider a random variable $X$ with a uniform distribution. It may seem irrational to consider a uniform distribution as the quantizer itself is optimized for a uniform distribution. However, as discussed previously, the mismatch could still affect the quantization performance significantly. We assume that $X \sim \mathcal{U}[\theta_1,\theta_2]$. In this case, the PDF of $X_{a,\lambda}$ is given as (see Appendix~\ref{subsec:modulo_folded_uniform_dist})
\begin{align}
    \label{eq:f_X_mod_uniform}
    f_{X_{a, \lambda}}(\theta) = \begin{cases}
        \beta \, {\mathbbm{1}}_{[\theta_{1\lambda},\theta_{2\lambda}]}(\theta), & \left\lfloor \frac{a\theta_2 + \lambda}{2\lambda} \right\rfloor = \left\lfloor \frac{a\theta_1 + \lambda}{2\lambda} \right\rfloor,\\
        \frac{\beta}{a} \left( {\mathbbm{1}}_{[\theta_{1\lambda},\lambda)}(\theta) + {\mathbbm{1}}_{[-\lambda,\theta_{2\lambda}]}(\theta) + \left\lfloor \frac{a\theta_2 + \lambda}{2\lambda} \right\rfloor - \left\lfloor \frac{a\theta_1 + \lambda}{2\lambda} \right\rfloor - 1 \right), & \left\lfloor \frac{a\theta_2 + \lambda}{2\lambda} \right\rfloor \neq \left\lfloor \frac{a\theta_1 + \lambda}{2\lambda} \right\rfloor,
    \end{cases} 
\end{align} 
where $\beta = \frac{1}{\theta_2 - \theta_1}$, $\theta_{1\lambda} = \mathcal{M}_{a,\lambda}(\theta_1)$, $\theta_{2\lambda} = \mathcal{M}_{a,\lambda}(\theta_2)$ and $\mathbbm{1}_{A}(x)$ is the indicator function defined as
\begin{align}
    \label{eq:indicator_definition_main}
     {\mathbbm{1}}_{A}(x) = \begin{cases}
        1, & \quad x \in A, \\
        0, & \quad \text{otherwise}.
    \end{cases}
\end{align}  
Under the condition of a large $\frac{a(\theta_1 - \theta_2)}{\lambda}$, we showed that $f_{X_{a, \lambda}}$ (cf. Appendix~\ref{subsec:modulo_folded_uniform_dist}), as described in \eqref{eq:f_X_mod_uniform}, approaches the density of a uniform distribution over $[-\lambda, \lambda]$. Note that the standard deviation of a random variable $X \sim \mathcal{U}[\theta_1, \theta_2]$ is given by $\frac{1}{\theta_1 - \theta_2}$. Thus, as observed with the Gaussian and Exponential distributions, the approximation improves as \(a\) increases relative to the standard deviation of the distribution.


In a nutshell, we showed that for the Gaussian, exponential, and uniform distributions, as the amplitude factor increases compared to the standard deviations of the distribution, the folded variables' distributions match that of the quantizer. Hence, the issue of quantizer mismatch is eradicated. We also demonstrated that, among quantizers with the same amplitude factor, the folded variable best approximates the quantizer distribution when the quantizer has a smaller folding parameter.

\section{Conclusion}
In this work, we proposed an approach to reduce the quantization mismatch error. Our proposed blind and adaptive method effectively minimizes distribution mismatch between input signals and quantizers without requiring prior knowledge of the input distribution. This approach is particularly successful with Gaussian, exponential, and uniform inputs. Although modulo-folding's many-to-one nature necessitates using existing unfolding techniques and considerable oversampling, the substantial reduction in mismatch and quantization error offers a worthwhile trade-off akin to those found in predictive coding strategies.

\appendices
\section{PDF of Modulo Folded Random Variable}
\label{sec:interchanging _sum_and_derivative}
\begin{theorem}
Consider a random variable $X$ with CDF $F_X$ and PDF $f_X$. Let ${X_{a,\lambda}}(\theta)$ represent the modulo-folded random variable defined by
\begin{align}
x_{a,\lambda}[n] = \mathcal{M}_{a,\lambda}(x[n]) = (a \,x[n] + \lambda) \,\, \text{mod}\,\, 2\lambda - \lambda.
\label{eq:scale_modulo_appendix}
\end{align}
This variable has a CDF ${F}_{X_{a,\lambda}}(\theta)$ defined for $\theta \in (-\lambda, \lambda)$. If the function $f_X(x)$ is continuous, approaches zero as $|x| \rightarrow \infty$, and is bounded, then for $\theta \in (-\lambda, \lambda)$, the following equation holds:
\begin{align}
\frac{\mathrm{d}}{\mathrm{d}\theta} {F}_{X_{a,\lambda}}(\theta) = \frac{1}{a}\sum_{m=-\infty}^{\infty} f_X \left( \frac{2m\lambda}{a} + \frac{\theta}{a} \right).
\end{align}
\end{theorem}
\begin{proof}
We know the CDF $f_{X_{\lambda}}$ is
 \begin{align}
    \label{eq:cdf_X_lambda_app}
        F_{X_{a,\lambda}}(\theta) = \sum_{m=-\infty}^{\infty} \left[{F}_{X}\left(\frac{2m\lambda}{a} + \frac{\theta}{a}\right) - {F}_{X}\left((2m-1)\frac{\lambda}{a}\right)\right].
    \end{align}
for $\theta \in [-\lambda,\lambda)$
Our result is trivial if
\begin{align}
    \label{eq:cdf_X_lambda_differentiated}
    \frac{\mathrm{d}}{\mathrm{d}\theta} \left( \sum_{m=-\infty}^{\infty} 
    \left[{F}_{X}\left(\frac{2m\lambda}{a} + \frac{\theta}{a}\right) - {F}_{X}\left((2m-1)\frac{\lambda}{a}\right)\right]\right) = \sum_{m=-\infty}^{\infty} \frac{\mathrm{d}}{\mathrm{d}\theta} \left[{F}_{X}\left(\frac{2m\lambda}{a} + \frac{\theta}{a}\right) - {F}_{X}\left((2m-1)\frac{\lambda}{a}\right)\right].
\end{align}
To show this, we employ the theorem for interchanging summation and differentiation as stated in \cite[Thm 7.17]{rudin1964principles}. For completeness, we reproduced the theorem as follows.
\begin{theorem}[\cite{rudin1964principles} ]
    Suppose $\{g_n\}$ is a sequence of functions, each differentiable on $[a,b]$, and $\{g_n(x_0)\}$ converges for some $x_0 \in [a,b]$. If the sequence of derivatives $\{g^{\prime}_n\}$ converges uniformly on $[a,b]$, then $\{g_n\}$ converges uniformly on $[a,b]$ to a function $g$, and 
\begin{align}
    g^{\prime}(x) = \lim_{n \rightarrow \infty}g^{\prime}_n(x) \quad \text{for all } x \in [a,b].
\end{align}
\end{theorem}
To apply this theorem to our case, we define the sequence of partial sums as
\begin{align}
    g_n(\theta) = \sum_{m=-n}^{n} \left({F}_{X}\left(\frac{2m\lambda}{a} + \frac{\theta}{a}\right) - {F}_{X}\left((2m-1)\frac{\lambda}{a}\right)\right). \label{eq:gn}
\end{align}
Then, for the interchange of the derivative and the summation, the following conditions have to hold.
\begin{enumerate}
    \item Convergence of the sum: $\lim_{n \rightarrow \infty} g_n(\theta) = g(\theta)$.
    \item Existence of the derivative: $g_n^{\prime}(\theta)$ should exist for all $n$.
     \item Uniform convergence of the derivative: $\lim_{n \rightarrow \infty} g_n^{\prime}(\theta) = g(\theta)^{\prime}$.
\end{enumerate}

From the definition of the CDF of the modulo variable, the first condition is true. Next, since $f_X$ is continuous and bounded, the corresponding CDF, $F_X$, is differentiable. Hence, we have that
\begin{align}
    g^{\prime}_n(\theta) = \frac{1}{a}\sum_{m=-n}^{n} f_X \left( \frac{2m\lambda}{a} + \frac{\theta}{a} \right).
\end{align}
Next, we show that the series $g_n^{\prime}(\theta)$ converges uniformly.
To this end, we apply the following lemma to demonstrate the uniform convergence of the series of derivatives.
\begin{lemma}
  If the PDF $f_X(x)$ is continuous and approaches zero as $|x| \to \infty$ (as is the case for distributions like the normal distribution or others with fast-decaying tails), then the series of derivatives  $g^{\prime}_n(\theta)$, given by
  \begin{align}
    \label{eq:DUTSS_derivative_series}
      g^{\prime}_n(\theta) = \frac{1}{a}\sum_{m=-n}^{n} f_X \left( \frac{2m\lambda}{a} + \frac{\theta}{a} \right),
  \end{align}
  converges uniformly.
\end{lemma}

\begin{proof}
To show uniform convergence, we apply the Cauchy criterion for uniform convergence (\cite[Thm 7.8]{rudin1964principles}). Specifically, we need to show that for every \( \epsilon > 0 \), there exists an \( N \) such that for all \( n, n' \geq N \) and for all \( \theta \)
    \begin{align}
    \left|\ g^{\prime}_n(\theta) - g^{\prime}_{n'}(\theta) \right| < \epsilon.    
    \end{align}
    Consider
    \begin{align}
        \left| g^{\prime}_n(\theta) - g^{\prime}_{n'}(\theta) \right| = \left| \frac{1}{a} \sum_{m=-n}^{n} f_X\left( \frac{2m\lambda}{a} + \frac{\theta}{a} \right) - \frac{1}{a} \sum_{m=-n'}^{n'} f_X\left( \frac{2m\lambda}{a} + \frac{\theta}{a} \right) \right|.
    \end{align}
    Without loss of generality, assume $ n' > n $. Then:
    \begin{align}
    \label{eq:gn_prime_diff_n_n'}
    \left| g^{\prime}_n(\theta) - g^{\prime}_{n'}(\theta) \right| = \frac{1}{a} \left| \sum_{m=n+1}^{n'} f_X\left( \frac{2m\lambda}{a} + \frac{\theta}{a} \right) + \sum_{m=-n'}^{-n-1} f_X\left( \frac{2m\lambda}{a} + \frac{\theta}{a} \right) \right|.
    \end{align}
Define $h(\tau) = f_X\left( \frac{2\tau\lambda}{a} + \frac{\theta}{a} \right)$. Then, we have
\begin{align}
\int_{-\infty}^{\infty} h(\tau) \, \mathrm{d}\tau = \int_{-\infty}^{\infty} f_X\left( \frac{2\tau\lambda}{a} + \frac{\theta}{a} \right) \, \mathrm{d}\tau = \frac{a}{2\lambda}
 \implies \int_{n+1}^{\infty} h(\tau) \, \mathrm{d}\tau < \frac{a}{2\lambda}.
\end{align}
Since $\int_{n+1}^{\infty} h(\tau) \, \mathrm{d}\tau$ converges to a value within $\left[0, \frac{a}{2\lambda}\right]$, it satisfies the integral test for convergence. Therefore, the series $\sum_{m=n+1}^{\infty} h(m)$ also converges. 
Thus, we have
\begin{align}
\sum_{m=n+1}^{\infty} f_X\left( \frac{2m\lambda}{a} + \frac{\theta}{a} \right) = k
\end{align} where $k>0$. As $n \to \infty$, $\sum_{m=n+1}^{\infty} f_X\left( \frac{2m\lambda}{a} + \frac{\theta}{a} \right) \to 0.$
Therefore, for some $N$, we have
\begin{align}
\label{eq:sample_sum_upper}
\sum_{m=n+1}^{\infty} f_X\left( \frac{2m\lambda}{a} + \frac{\theta}{a} \right) < \frac{a\epsilon}{2} \quad \forall\;n > N.
\end{align}
Similarly, it can be shown that 
\begin{align}
\label{eq:sample_sum_lower}
\sum_{m=-\infty}^{-n-1} f_X\left( \frac{2m\lambda}{a} + \frac{\theta}{a} \right) < \frac{a\epsilon}{2} \quad \forall\;n > N.
\end{align}
    Using \eqref{eq:sample_sum_lower} and \eqref{eq:sample_sum_upper} in \eqref{eq:gn_prime_diff_n_n'} we have,
    \begin{align}
         \left| g^{\prime}_n(\theta) - g^{\prime}_{n'}(\theta) \right| < \epsilon.
    \end{align}
    This shows that \( g^{\prime}_n(\theta) \) satisfies the Cauchy criterion, and therefore, the series of derivatives $g_n^{\prime}$ converges uniformly.
\end{proof}
    We have shown that the conditions for interchanging differentiation and summation are met, and thus we have

   \begin{align}
      g^{\prime}(\theta) &= \lim_{n \to \infty}\:g^{\prime}_n(\theta)\\
      \frac{\mathrm{d}}{\mathrm{d}\theta} F_{X_{a,\lambda}}(\theta) &= \lim_{n \to \infty} \frac{1}{a}\sum_{m=-n}^{n} f_X \left( \frac{2m\lambda}{a} + \frac{\theta}{a} \right) \\
      \frac{\mathrm{d}}{\mathrm{d}\theta} F_{X_{a,\lambda}}(\theta) &= \frac{1}{a} \sum_{m=-\infty}^{\infty} f_X \left( \frac{2m\lambda}{a} + \frac{\theta}{a} \right) 
   \end{align}

\end{proof}

\section{Second Wasserstein Distances}
\label{app:SWD}
In this section, we will either calculate or present the precomputed Second Wasserstein Distance between several commonly used distributions. The Second Wasserstein Distance between two random variables $X$ and $Y$ is given by:

\begin{align}
    W_2(X, Y) = \left( \int_0^1 \left( F_X^{-1}(u) - F_Y^{-1}(u) \right)^2 \, du \right)^{\frac{1}{2}},
\end{align}

where $F_X^{-1}$ and $F_Y^{-1}$ denote the inverse CDFs of $X$ and $Y$, respectively. The expressions for the distances for some distributions are given in certain textbooks; here, we derived them for completeness.

\subsection{Gaussian Distributions}
The second Wasserstein distance between Gaussian distributions has an analytical solution, as detailed in \cite{givens1984class}. The general solution is as follows Let \( X_1 \) and \( X_2 \) be two Gaussian distributions with mean vectors \( \bm{\mu_1} \) and \( \bm{\mu_2} \), and non-singular covariance matrices \( \Sigma_1 \) and \( \Sigma_2 \), respectively. The \( L^2 \) Wasserstein distance \( W_2(X_1, X_2) \) is given by

\begin{align}
    \label{eq:general_gaussian_wasserstein}
    \sqrt{\|\bm{\mu_1}-\bm{\mu_2}\|^2 + \text{tr}(\Sigma_1 + \Sigma_2 - 2\sqrt{\Sigma_1}\,\Sigma_2\sqrt{\Sigma_1})}.
\end{align}

For one-dimensional distributions, where \( \Sigma_1 = \sigma_1^2 \) (with \( \sigma_1 \) being the standard deviation of \( X_1 \)) and similarly \( \Sigma_2 = \sigma_2^2 \) for \( X_2 \), the equation simplifies to

\begin{align}
    W_2(X_1, X_2) = \sqrt{(\mu_2 - \mu_1)^2 + (\sigma_2 - \sigma_1)^2}.
\end{align}

\subsection{Exponential and Uniform Distributions}

Let \( X \) be a random variable following an exponential distribution \(\exp(p)\) and \( Y \) be a random variable with a uniform distribution over \([0, C]\). Their CDFs are given as
\begin{align}
    F_X(x) = \begin{cases}
      1 - e^{-px}, & x \geq 0, \\
      0, & x < 0,
    \end{cases} \quad \text{and} \quad
    F_Y(y) = \begin{cases}
      1, & y > C, \\
      \frac{y}{C}, & y \in [0, C], \\
      0, & y < 0,
    \end{cases}
\end{align}
respectively.  Their inverse CDFs, or quantile functions, are given as
\begin{align}
    F_X^{-1}(u) = -\frac{\ln(1-u)}{p} \quad \text{and} \quad
    F_Y^{-1}(u) = Cu,
\end{align}
respectively.
By using these, the second Wasserstein distance between \( X \) and \( Y \) is then calculated as
\begin{align}
    W_2(X,Y)^2 &= \int_0^1 \left| F_X^{-1}(u) - F_Y^{-1}(u) \right|^2 du \nonumber\\
    &= \int_0^1 \left| -\frac{ln(1-u)}{p} - Cu \right|^2 du \nonumber\\
    &= \int_0^1 \left| -\frac{ln(u)}{p} - C(1-u) \right|^2 du \nonumber\\
    &= \int_0^1 \left( \left( \frac{ln(u)}{p} \right)^2 + 2C\frac{u\:ln(1-u)}{p}+ C^2(1-u)^2 \right) du \nonumber\\
    W_2(X,Y) &= \sqrt{\frac{2}{p^2} - \frac{3C}{2p}+ \frac{C^2}{3}}.
\end{align} 

\subsection{Exponential Distributions}
Consider two random variables \( X \) and \( Y \), each following an exponential distribution, \(\exp(p_1)\) and \(\exp(p_2)\) respectively. The inverse CDFs, or quantile functions, for these distributions, respectively, are given by
\begin{align}
    F_X^{-1}(u) = -\frac{\ln(1-u)}{p_1}, \quad \text{and} \quad
    F_Y^{-1}(u) = -\frac{\ln(1-u)}{p_2}.
\end{align}
The second Wasserstein distance between \( X \) and \( Y \) is calculated as 
\begin{align}
    W_2(X,Y)^2 &= \int_0^1 \left| F_X^{-1}(u) - F_Y^{-1}(u) \right|^2 du \\
    &= \int_0^1 \left| \frac{\ln(1-u)}{p_1} - \frac{\ln(1-u)}{p_2} \right|^2 du \\
    &= \left| \frac{1}{p_1} - \frac{1}{p_2} \right|^2 \int_0^1 \ln(1-u)^2 \, du \\
    &= 2 \left| \frac{1}{p_1} - \frac{1}{p_2} \right|^2\\
    W_2(X,Y) &= \sqrt{2} \left| \frac{1}{p_1} - \frac{1}{p_2} \right|.
\end{align}

\section{Probability Distribution after Modulo Folding}

We define the modulo folding \(x_{a,\lambda}\) of a signal \(x\) as:
\begin{align}
    x_{a,\lambda}[n] = \mathcal{M}_{a,\lambda}(x[n])
     = (a \,x[n]+\lambda)\,\, \text{mod}\,\, 2\lambda -\lambda.
    \label{eq:scale_modulo_copy}
\end{align}

For the random variable $X$ with PDF $f_X$ and CDF $F_X$, let the PDF of the folded random variable $X_{a, \lambda}$ be denoted as $f_{X_{a,\lambda}}$. Then the CDF $F_{X_{a, \lambda}}$ is determined as

\begin{align}
    \label{eq:cdf_modulo_folded}
    F_{X_{a, \lambda}} (\theta) = \sum_{m=-\infty}^{\infty} 
    {F}_{X}\left( \frac{2m\lambda}{a} + \frac{\theta}{a} \right) - {F}_{X}\left((2m-1)\frac{\lambda}{a}\right).
\end{align}

If $f_X(x)$ is continuous, approaches zero as $|x| \rightarrow \infty$, and is bounded, then the PDF of $X_{a,\lambda}$ is given as

\begin{align}
    \label{eq:pdf_modulo_folded}
        f_{X_{a, \lambda}} (\theta) = \frac{1}{a} \sum_{m=-\infty}^{\infty} 
    {f}_{X}\left(\frac{2m\lambda}{a} + \frac{\theta}{a}\right).
\end{align}
Next, we determine the distribution for the folded variable when $X$ follows a Gaussian distribution.
\subsection{Gaussian Distribution}
\label{subsec:modulo_folded_gaussian_dist}

For a Gaussian random variable with mean \( \mu \) and standard deviation \( \sigma \), the PDF \( f_X(x) \) is expressed as
\begin{align}
    f_X(x) = \frac{1}{\sqrt{2 \pi \sigma^2}} e^{-\frac{(x-\mu)^2}{2\sigma^2}}.
\end{align}
Now, consider the modulo folded signal $X_{a, \lambda}$, which has a probability density function $f_{X_{a,\lambda}}$. This function can be represented as:
\begin{align}
    \label{eq:gaussian_fx_sum}
    f_{X_{a,\lambda}}(\theta) = \frac{1}{a} \sum_{m=-\infty}^{\infty} 
    {f}_{X}\left(\frac{2m\lambda}{a} + \frac{\theta}{a}\right) = \sum_{m=-\infty}^{\infty}\frac{1}{\sqrt{2 \pi \sigma^2}} e^{-\frac{(2m\lambda + \theta - a\mu)^2}{2(a\sigma)^2}}.
\end{align}
Since $X_{a, \lambda} \in [-\lambda, \lambda]$, to show that the density of the folded variable tends to uniform distribution for large $a$, we need to prove that  $f_{X_{a,\lambda}}(\theta) \rightarrow \frac{1}{2\lambda}$. Since the summation in \eqref{eq:gaussian_fx_sum} can not be simplified further, we consider an alternative approach, as discussed next.


We first define a variable $w_m$ as 
\begin{align}
    w_m = \frac{2m\lambda + \theta - a\mu}{a\sigma}.
\end{align}
Note that the difference between consecutive $w_m$s, denoted $\Delta w_m$, is given by\begin{align}
    \Delta w_m = w_m - w_{m-1} = \frac{2\lambda}{a\sigma}.
\end{align}
Substituting $w_m$ in \eqref{eq:gaussian_fx_sum}, the summation is given as
\begin{align}
\label{eq:Gaussian_Reimman_Sum}
     \frac{1}{a} \sum_{m=-\infty}^{\infty} \frac{1}{\sqrt{2 \pi \sigma^2}} e^{-\frac{w_m^2}{2}}  = \frac{1}{2\sqrt{2\pi}\lambda} \sum_{m=-\infty}^{\infty} e^{-\frac{w_m^2}{2}}\, \Delta w_m.
\end{align}
As $a\rightarrow \infty$, $\Delta w_m$ tends to zero. Hence, in the limit, the sum is transformed into a Riemann integral as
\begin{align}
    \label{eq:Gaussian_Reimman_Integral}
    \lim_{\Delta w_m \rightarrow 0}\frac{1}{2\sqrt{2\pi}\lambda} \sum_{m=-\infty}^{\infty} e^{-\frac{w_m^2}{2}}\, \Delta w_m &= \frac{1}{2\sqrt{2\pi}\lambda} \int_{-\infty}^{\infty} e^{-\frac{u^2}{2}} \, du = \frac{1}{2\lambda}.
\end{align}
This shows that for small $\frac{\lambda}{a\sigma}$, the probability distribution of $ f_{X_{a,\lambda}} $  approaches a uniform distribution on modulo folding.
\subsection{Exponential Distribution}
\label{subsec:modulo_folded_exponential_dist}

Consider an exponential random variable $X$ with a rate parameter $p$. The PDF of $X$, denoted $f_X(x)$ is given by
\begin{align}
    \label{eq:pdf_exponential}
    f_X(x) = \begin{cases}
        pe^{-px}, & x \geq 0, \\
        0, & x < 0,
    \end{cases}
\end{align}
and its CDF $F_X(x)$ is expressed as
\begin{align}
    \label{eq:cdf_exponential}
    F_X(x) = \begin{cases}
        1 - e^{-px}, & x \geq 0, \\
        0, & x < 0.
    \end{cases}
\end{align}
Now, consider the modulo folded signal $ {X_{a,\lambda}} $, which has a CDF denoted by $ F_{X_{a,\lambda}}(\theta) $. From \eqref{eq:cdf_modulo_folded}, we have:
\begin{align}
    \label{eq:cdf_X_lambda_exp_definition}
        F_{X_{a, \lambda}} (\theta) = \sum_{m=-\infty}^{\infty} 
    {F}_{X}\left( \frac{2m\lambda}{a} + \frac{\theta}{a} \right) - {F}_{X}\left((2m-1)\frac{\lambda}{a}\right)
\end{align} for $\theta \in [-\lambda,\lambda)$.
Since the CDF $F_X(x)$ depends significantly on whether $x > 0$ or not, we need to analyze the signs of $\left( \frac{2m\lambda}{a} + \frac{\theta}{a} \right)$ and $(2m-1)\frac{\lambda}{a}$. In particular, we consider two cases, $\theta \geq 0$ and $\theta <0$, to determine $F_{X_{a, \lambda}} (\theta)$.
\textbf{Case 1: $\theta < 0$}\\
When $\theta < 0$, $\left( \frac{2m\lambda}{a} + \frac{\theta}{a} \right) > 0$ for $m > 0$, and $(2m-1)\frac{\lambda}{a} > 0$ for $m > 0$. Thus, equation \eqref{eq:cdf_X_lambda_exp_definition} simplifies to
    \begin{align}
        F_{X_{a, \lambda}} (\theta) &= \sum_{m=1}^{\infty} \left(
    {F}_{X}\left( \frac{2m\lambda}{a} + \frac{\theta}{a} \right) - {F}_{X}\left((2m-1)\frac{\lambda}{a}\right) \right)\\
          &= \sum_{m=1}^{\infty} 
        \left( (1-e^{-\frac{p}{a}(2m\lambda + \theta)}) - (1-e^{-\frac{p}{a}(2m-1)\lambda}) \right) \nonumber \\
        &= \frac{e^{\frac{p}{a}\lambda} - e^{-\frac{p}{a}\theta}}{e^{2\frac{p}{a}\lambda}-1}.  \label{eq:exp_sum_simplified}
    \end{align}
\textbf{Case 2: $\theta \geq 0$}\\
When $\theta \geq 0$, $\left( \frac{2m\lambda}{a} + \frac{\theta}{a} \right) > 0$ for $m \geq 0$, and $(2m-1)\frac{\lambda}{a} > 0$ for $m > 0$. Thus, equation \eqref{eq:cdf_X_lambda_exp_definition} simplifies to
\begin{align}
        F_{X_{a, \lambda}} (\theta) &= \sum_{m=1}^{\infty} \left(
    {F}_{X}\left( \frac{2m\lambda}{a} + \frac{\theta}{a} \right) - {F}_{X}\left((2m-1)\frac{\lambda}{a}\right) \right) + F_X\left(\frac{\theta}{a}\right) = \left( \frac{e^{\frac{p}{a}\lambda} - e^{-\frac{p}{a}\theta}}{e^{2\frac{p}{a}\lambda}-1} \right) + 1-e^{-\frac{p}{a}\theta}. \nonumber
\end{align}
Thus, the CDF of the exponential random variable after modulo folding is given by 
\begin{align}
    \label{cdf_modulofolded_exponential}
    F_{X_{a, \lambda}} (\theta) = \begin{cases}
        \frac{e^{\frac{p}{a}\lambda} - e^{-\frac{p}{a}\theta}}{e^{2\frac{p}{a}\lambda}-1} + 1 - e^{-\frac{p}{a}\theta}, & \theta \in [0,\lambda),\\
        \frac{e^{\frac{p}{a}\lambda} - e^{-\frac{p}{a}\theta}}{e^{2\frac{p}{a}\lambda}-1}, & \theta \in [-\lambda,0).
    \end{cases}
\end{align}
Since $ {X_{a, \lambda}} \in [-\lambda,\lambda)$, to show that its distribution tends to uniform distribution for large $a$, we need to show that $ F_{X_{a, \lambda}}(\theta) \rightarrow \frac{\theta + \lambda}{2\lambda}$.
Using the approximation $e^{x} \approx 1 + x$  for very small  $x$, \eqref{cdf_modulofolded_exponential} becomes 
\begin{align}
    F_{X_{a, \lambda}} (\theta)  = \frac{\theta + \lambda}{2\lambda}, \quad \theta \in [-\lambda, \lambda).
\end{align}



Therefore, when $\frac{p}{a}$ is small, a random variable following an exponential distribution converges to a uniform distribution in distribution after modulo folding.

\subsection{Uniform Distribution}
\label{subsec:modulo_folded_uniform_dist}
Consider a uniform random variable $X$ defined over the interval $[\theta_1, \theta_2]$. The probability density function $f_X$ for this random variable is expressed as 
\begin{align}
    \label{eq:f_X_uniform_def}
    f_X(\theta) = \begin{cases}
        \frac{1}{\theta_2 - \theta_1} = \beta, & \theta \in [\theta_1, \theta_2],\\
        0, & \text{otherwise.}
    \end{cases}
\end{align}

We denote the PDF of the modulo-folded version as $f_{X_{a, \lambda}}$.
If $\theta_1 = -\theta_2$, selecting $a =  \frac{\lambda}{\theta_2}$ yields the desired uniform distribution $\mathcal{U}[-\lambda, \lambda]$. However, the distribution of $X$ is not known, and scaling is not possible. We show that for large values of $a$, the distribution $f_{X_{a, \lambda}}$ tends to $\mathcal{U}[-\lambda, \lambda]$.
We define the following parameters that are used in the subsequent derivations:
\begin{align}
    \theta_{1\lambda} = (a\theta_1 + \lambda) \,\, \text{mod} \,\, 2\lambda - \lambda, \quad \text{and} \quad
    \theta_{2\lambda} = (a\theta_2 + \lambda) \,\, \text{mod} \,\, 2\lambda - \lambda.
\end{align}
From Equation \eqref{eq:pdf_X_mod}, the PDF of the modulo-folded distribution is given by
\begin{align}
    f_{X_{a,\lambda}}(\theta) = \frac{1}{a}\sum_{m=-\infty}^{\infty} f_{X}\left(\frac{2m\lambda}{a} + \frac{\theta}{a}\right).
\end{align}
\begin{figure}[!t]
    \centering
    \subfigure[$f_X(\theta)$]{\includegraphics[width = 3.2 in]{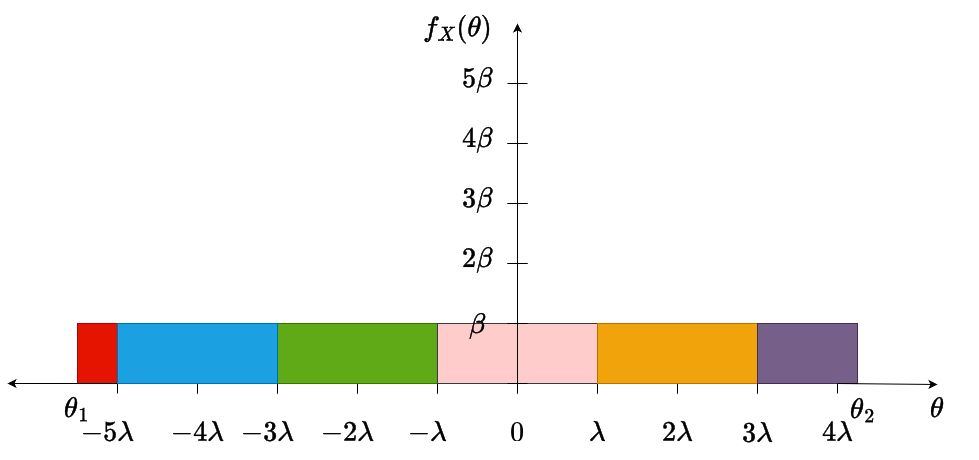}
    \label{fig:uniform_premodulo}
    }
    \subfigure[$f_{X_{1, \lambda}}$]{\includegraphics[width = 3.2 in]{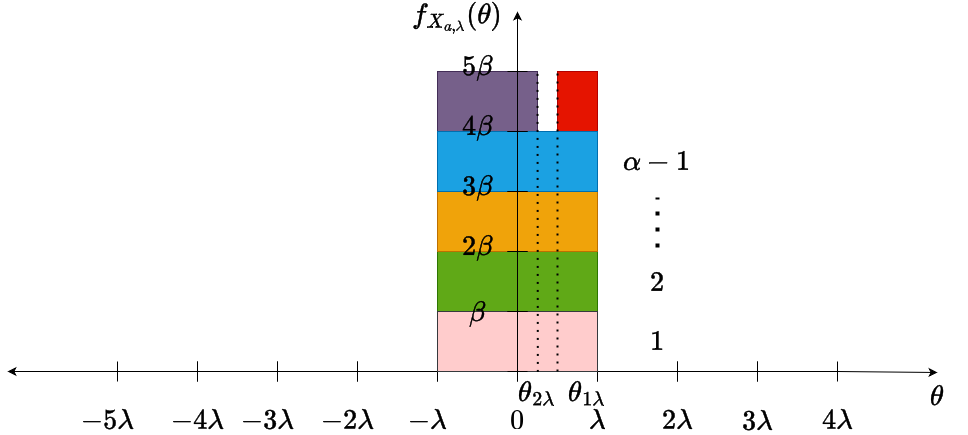}
           \label{fig:uniform_postmodulo}
           }
    \caption{Distribution (PDF) transformation of $X$ under modulo folding with $a = 1$}
    \label{fig:distibution_change_uniform}
\end{figure}
Fig.~\ref{fig:distibution_change_uniform} illustrates the transformation of a uniform random variable $f_X$ to $f_{X_{a,\lambda}}$ upon modulo-folding for $a=1$. The distribution of the folded variable is derived from the contributions of three primary components from the original distribution:

\begin{enumerate}
    \item The leftmost segment starting from $\theta_1$, which is mapped from $\theta_{1\lambda}$ to $\lambda$.
    \item The rightmost segment extending to $\theta_2$, which is mapped from $-\lambda$ to $\theta_{2\lambda}$.
    \item All central segments, which effectively add a constant value over the range $[-\lambda,\lambda)$.
\end{enumerate}
We will now demonstrate this through detailed calculations.\\
Since $\theta \in [-\lambda, \lambda)$ and $f_X(\theta)$ is supported over $[\theta_1, \theta_2]$, only a finite terms in the summation will contribute to $ f_{X_{a,\lambda}}(\theta)$. In particular, we have that
\begin{align}
    f_{X_{a,\lambda}}(\theta) = \frac{1}{a}\sum_{m=m_1}^{m_2} f_{X}\left(\frac{2m\lambda}{a} + \frac{\theta}{a}\right), 
\end{align}
where 
\begin{align}
    m_1 = \left\lfloor \frac{a\theta_1 + \lambda}{2\lambda} \right\rfloor \quad \text{and} \quad  m_2 = \left\lfloor \frac{a\theta_2 + \lambda}{2\lambda} \right\rfloor.
\end{align}
To establish a relationship between $\theta_1$, $\theta_{1\lambda}$, and $m_1$, we express $a\theta_1 + \lambda$ as the product of $2\lambda$ and its quotient, plus the remainder (Division Algorithm). Since $\lambda > 0$, we have
\begin{align}
    a\theta_1 + \lambda &= 2\lambda \left\lfloor \frac{a\theta_1 + \lambda}{2\lambda} \right\rfloor +   (a\theta_1 + \lambda) \,\, \text{mod} \,\, 2\lambda\\
   \implies  a\theta_1 + \lambda &= 2\lambda m_1 +   \theta_{1\lambda}+\lambda \nonumber\\ 
    \label{eq:theta_1_m_1}
    a\theta_1 &= 2\lambda m_1 +   \theta_{1\lambda}.
\end{align}
Similarly, we have
\begin{align}
    \label{eq:theta_2_m_2}
    a\theta_2 = 2\lambda m_2 +   \theta_{2\lambda}.
\end{align}
The indicator function is defined as
\begin{align}
    \label{eq:g_ab_definition_appendix}
    {\mathbbm{1}}_{A}(x) = \begin{cases}
        1, & \quad x \in A, \\
        0, & \quad \text{otherwise}.
    \end{cases}
\end{align}
When $m_1 = m_2 = m^{\prime}$, $f_{X_{a,\lambda}}(\theta)$ is trivial and is given by
\begin{align}
    \label{eq:f_X_uniform_m1_m2_equal_def}
    f_{X_{a,\lambda}}(\theta) &= \frac{1}{a}f_{X}\left(\frac{2m^{\prime}\lambda}{a} + \frac{\theta}{a}\right).
\end{align}
Using \eqref{eq:theta_1_m_1} and \eqref{eq:theta_2_m_2} we have
\begin{align}
    \frac{2m^{\prime}\lambda + \theta}{a} = \theta_1 + \frac{\theta -\theta_{1\lambda}}{a} = \theta_2 + \frac{\theta -\theta_{2\lambda} }{a},
\end{align}
which we use to find $\theta$ for which $\frac{2m^{\prime}\lambda + \theta}{a} \in [\theta_1,\theta_2]$. Thus, we have
\begin{align}
    \frac{2m^{\prime}\lambda + \theta}{a} \geq \theta_1 \iff \theta \geq \theta_{1\lambda} &\quad \text{and} \quad \frac{2m^{\prime}\lambda + \theta}{a} \leq \theta_2 \iff \theta  \leq \theta_{2\lambda} \nonumber\\
    \implies \frac{2m^{\prime}\lambda + \theta}{a} \in [\theta_1,\theta_2] &\iff \theta \in [\theta_{1\lambda},\theta_{2\lambda}].
\end{align}
Using this in \eqref{eq:f_X_uniform_def} and \eqref{eq:f_X_uniform_m1_m2_equal_def} we have
\begin{align}
    \label{eq:f_X_uniform_m1_m2_equal}
    f_{X_{a,\lambda}}(\theta) &= \beta {\mathbbm{1}}_{[\theta_{1\lambda},\theta_{2\lambda}]}.
\end{align}
Next, we consider when $m_1 \neq m_2$. We define $\alpha$ to be $m_2 - m_1$ to help calculate $f_{X_{a,\lambda}}(\theta)$ for different $m$-s. 
Lets now analyze the function $f_{X}\left(\frac{2m^{\prime}\lambda}{a} + \frac{\theta}{a}\right)$ :\\
\textbf{Case 1:} $m = m_1$
\begin{align}
f_{X}\left(\frac{2m_1\lambda}{a} + \frac{\theta}{a} \right) = 
\begin{cases}
\beta, & 2 m_1 \lambda + \theta \in [a\theta_1, a\theta_2], \\
0, & \text{otherwise}.
\end{cases}
\end{align}
The condition $2 m_1 \lambda + \theta \in [a\theta_1, a\theta_2]$ can be also expressed as 
\begin{align}
 \theta \in [a\theta_1-2 m_1 \lambda, a\theta_2-2 m_1 \lambda].
\end{align}
Using equation \eqref{eq:theta_1_m_1}, we get
\begin{align}
\theta \in [\theta_{1\lambda}, a\theta_2 - 2 m_1 \lambda].
\end{align}
Next, consider the relationship for $\theta_2$
\begin{align}
a\theta_2 &= 2 m_2 \lambda + \theta_{2\lambda} \nonumber\\
\implies a\theta_2 - 2 m_1 \lambda &= 2(m_2 - m_1) \lambda + \theta_{2\lambda} \nonumber\\
a\theta_2 - 2 m_1 \lambda  &= 2 \alpha \lambda + \theta_{2\lambda}.
\end{align}
Knowing \(2 \alpha \lambda \geq 2 \lambda\) and \(\theta_{2\lambda} \geq -\lambda\), we can infer
\begin{align}
a\theta_2 - 2 m_1 \lambda \geq 2 \lambda - \lambda = \lambda.
\end{align}
We know $\theta \in [-\lambda, \lambda)$, this and the condition $\theta \in [\theta_{1\lambda}, a\theta_2 - 2 m_1 \lambda]$ implies
\begin{align}
 2 m_1 \lambda + \theta \in [a\theta_1, a\theta_2] \iff \theta \in [\theta_{1\lambda}, \lambda).
\end{align}
Thus, the function $f_{X}\left(\frac{2m_1\lambda}{a}+\frac{\theta}{a}\right)$ simplifies to
\begin{align}
f_{X}\left(\frac{2m_1\lambda}{a}+\frac{\theta}{a}\right) &= 
\begin{cases}
\beta, &  \theta \in [\theta_{1\lambda}, \lambda), \\
0, & \text{otherwise},
\end{cases}\\
&= \beta {\mathbbm{1}}_{[\theta_{1\lambda},\lambda)}(\theta).
\end{align}
\textbf{Case 2:} $m = m_2$
\begin{align}
f_{X}\left(\frac{2m_2\lambda}{a} + \frac{\theta}{a} \right) = 
\begin{cases}
\beta & \text{if } 2 m_2 \lambda + \theta \in [a\theta_1, a\theta_2] \\
0 & \text{otherwise}
\end{cases}
\end{align}
The condition $2 m_2 \lambda + \theta \in [a\theta_1, a\theta_2]$ can be also expressed as 
\begin{align}
 \theta \in [a\theta_1-2 m_2 \lambda, a\theta_2-2 m_2 \lambda].
\end{align}
Using equation \eqref{eq:theta_2_m_2}, we obtain
\begin{align}
\theta \in [ a\theta_1 - 2 m_2 \lambda,\theta_{2\lambda}].
\end{align}
Next, consider the following relationship for $\theta_1$
\begin{align}
a\theta_1 &= 2 m_1 \lambda + \theta_{1\lambda}\\
\implies a\theta_1 - 2 m_2 \lambda &= 2(m_1 - m_2) \lambda + \theta_{1\lambda}\nonumber\\
&= -2 \alpha \lambda + \theta_{1\lambda} .
\end{align}
Knowing $-2 \alpha \lambda \leq -2 \lambda$ and $\theta_{2\lambda} \leq \lambda$, we can infer
\begin{align}
a\theta_1 - 2 m_2 \lambda \leq -2 \lambda + \lambda = -\lambda.
\end{align}
We know $\theta \in [-\lambda, \lambda)$, this and the condition $\theta \in [ a\theta_1 - 2 m_2 \lambda,\theta_{2\lambda}]$ implies
\begin{align}
 2 m_1 \lambda + \theta \in [a\theta_1, a\theta_2] \iff \theta \in [- \lambda,\theta_{2\lambda}].
\end{align}
Thus, the function $f_{X}\left(\frac{2m_2\lambda}{a}+\frac{\theta}{a}\right)$ simplifies to
\begin{align}
f_{X}\left(\frac{2m_2\lambda}{a}+\frac{\theta}{a}\right) &= 
\begin{cases}
\beta, & \theta \in [- \lambda,\theta_{2\lambda}], \\
0, & \text{otherwise},
\end{cases}\\
&= \beta {\mathbbm{1}}{[-\lambda,\theta_{2\lambda}]}(\theta).
\end{align}
\textbf{Case 3:} $ m_1 < m < m_2$
\begin{align}
f_{X}\left(\frac{2m\lambda}{a} + \frac{\theta}{a} \right) = 
\begin{cases}
\beta, &  2 m \lambda + \theta \in [a\theta_1, a\theta_2], \\
0, & \text{otherwise}.
\end{cases}
\end{align}
To establish bounds for $2m\lambda+\theta$, we start by analyzing the following inequality
\begin{align}
    m &> m_1 \\
     \implies m &\geq m_1 + 1\nonumber\\
    2m\lambda &\geq 2 m_1 \lambda + 2\lambda.
\end{align}
Substituting for $m_1$ from \eqref{eq:theta_1_m_1}, we have
\begin{align}
    2m\lambda &\geq ( a{\theta_1} - {\theta_{1\lambda}}) + 2\lambda \nonumber\\
    2m\lambda + \theta &\geq  a{\theta_1} + 2\lambda - {\theta_{1\lambda}} + \theta.
\end{align}
Since both $\theta$ and $\theta_{1\lambda}$ lie in $[-\lambda, \lambda)$, we get 
\begin{align}
    \label{eq:lower_bound_2ml_theta}
    2m\lambda + \theta &>  a{\theta_1}
\end{align}
Next, to determine an upper bound for $m$, we consider
\begin{align}
    m &< m_2 \nonumber\\
    \implies m &\leq m_2 - 1 \nonumber\\
    2m\lambda &\leq 2 m_2 \lambda - 2\lambda.
\end{align}
Substituting for $m_2$ from \eqref{eq:theta_2_m_2}, we have
\begin{align}
    2m\lambda &\leq (a{\theta_2}  - {\theta_{2\lambda}}) - 2\lambda \nonumber \\
    2m\lambda + \theta &\leq  a{\theta_2} -2\lambda - {\theta_{2\lambda}} + \theta.
\end{align}
Since both $\theta$ and $\theta_{2\lambda}$ lie in $[-\lambda, \lambda)$, we get \begin{align}
    \label{eq:upper_bound_2ml_theta}
    2m\lambda + \theta &< a\theta_2.
\end{align}
From equations \eqref{eq:lower_bound_2ml_theta} and \eqref{eq:upper_bound_2ml_theta}, we establish that
\begin{align}
    \label{eq:bound_2ml_theta}
    2m\lambda+\theta  \in (a\theta_1,a\theta_2)
\end{align}
Since $2m\lambda+\theta$ is always within the support of $f_{X}\left(\frac{2m\lambda}{a} + \frac{\theta}{a} \right)$ we can say
\begin{align}
    f_{X}\left(\frac{2m\lambda}{a} + \frac{\theta}{a} \right) = \beta \quad \forall \; \theta \in [\-\lambda,\lambda).
\end{align}
Thus, for $m_1 \neq m_2$ we have
\begin{align}
f_{X}\left(\frac{2m\lambda}{a} + \frac{\theta}{a} \right) = 
\begin{cases}
\beta {\mathbbm{1}}_{[\theta_{1\lambda},\lambda)}(\theta), & m = m_1,\\
\beta, & m_1 < m < m_2,\\
\beta {\mathbbm{1}}_{[-\lambda,\theta_{2\lambda})}(\theta),
 & m = m_2.
\end{cases}
\end{align}
The PDF of the modulo-folded signal is then given by
\begin{align}
    f_{X_{a,\lambda}}(\theta) &= \frac{1}{a}\sum_{m=-\infty}^{\infty} f_{X}\left(\frac{2m\lambda}{a} + \frac{\theta}{a}\right) \nonumber\\
    &= \frac{1}{a}\sum_{m=m_1}^{m_2} f_{X}\left(\frac{2m\lambda}{a} + \frac{\theta}{a}\right) \nonumber\\
    &= \frac{1}{a}\left( \beta {\mathbbm{1}}_{[\theta_{1\lambda},\lambda)}(\theta) + (\alpha - 1) \beta + \beta {\mathbbm{1}}_{[-\lambda,\theta_{2\lambda}]}(\theta)  \right) \nonumber\\
    &= \frac{\beta}{a} \left( {\mathbbm{1}}_{[\theta_{1\lambda},\lambda)}(\theta) + {\mathbbm{1}}_{[-\lambda,\theta_{2\lambda}]}(\theta) + \alpha - 1 \right).
\end{align}
To show that the density of the folded variable tends to uniform distribution $\mathcal{U}[-\lambda,\lambda]$ for large $a$, we need to prove that  $f_{X_{a,\lambda}}(\theta) \rightarrow \frac{1}{2\lambda}$. 
We start by finding a bound on $\alpha$ from its definition,
\begin{align}
    \alpha &= m_2 - m_1 \\
    &= \left\lfloor \frac{a\theta_2 + \lambda}{2\lambda} \right\rfloor - \left\lfloor \frac{a\theta_1 + \lambda}{2\lambda} \right\rfloor.
\end{align}
Since $\lfloor x \rfloor \in (x-1,x]$ we can upper and lower bound as 
\begin{align}
    \frac{a\theta_2 - a\theta_1}{2\lambda} - 1 < \alpha <  \frac{a\theta_2 - a\theta_1}{2\lambda} + 1.
\end{align}
For very large values of $\frac{a(\theta_2 - \theta_1)}{\lambda}$ we have large $\alpha$, thus $\alpha$ can be approximated as 
\begin{align}
    \alpha \approx \frac{a(\theta_2 - \theta_1)}{2\lambda}.
\end{align}
Then, the PDF becomes
\begin{align}
    f_{X_{a,\lambda}}(\theta) &=\frac{\beta}{a} \left( {\mathbbm{1}}_{[\theta_{1\lambda},\lambda)}(\theta) + {\mathbbm{1}}_
    {[-\lambda,\theta_{2\lambda}]}(\theta) + \frac{a(\theta_2 - \theta_1)}{2\lambda} - 1 \right).
\end{align}
As $\frac{a(\theta_2 - \theta_1)}{\lambda}$ becomes larger, it dominates all the other terms to give
\begin{align}
    f_{X_{a,\lambda}}(\theta) &=\frac{\beta}{a} \left( \frac{a(\theta_2 - \theta_1)}{2\lambda} \right)\\
    f_{X_{a,\lambda}}(\theta) &=\frac{1}{2\lambda}.
\end{align}
Hence, for sufficiently large $\frac{a(\theta_2 - \theta_1)}{\lambda}$, the modulo-folded random variable ${X_{a,\lambda}}$ converges to a uniform random variable with distribution $\mathcal{U}[-\lambda,\lambda]$ in distribution.

\bibliographystyle{IEEEtran}
\bibliography{refs,refs1,refs2}

\end{document}